\documentclass{article}[fullpage,a4paper]

\usepackage{amsfonts,amsmath,amsthm,amssymb,braket,enumerate,mathtools,xcolor,float}
\usepackage[margin=2cm]{geometry}
\usepackage[backref]{hyperref}
\usepackage{comment}
\usepackage{todonotes}
\usepackage{authblk,nicefrac}

\usepackage{tikz,ifdraft,wrapfig}
\usetikzlibrary{quantikz2}
\usetikzlibrary{arrows,positioning,matrix,backgrounds,calc,fit}
\tikzcdset{
every matrix/.style={ampersand replacement=\&}
}

\usepackage{thmtools}
\usepackage{thm-restate}

\makeatletter
\renewcommand*\env@matrix[1][*\c@MaxMatrixCols c]{%
  \hskip -\arraycolsep
  \let\@ifnextchar\new@ifnextchar
  \array{#1}}
\makeatother

\newcommand{\C}{\mathbb{C}}
\newcommand{\I}{\mathbb{I}}
\newcommand{\N}{\mathbb{N}}

\newcommand{\R}{\mathbb{R}}
\newcommand{\T}{\mathbb{T}}
\newcommand{\Z}{\mathbb{Z}}
\newcommand{\crr}{\mathcal{R}}
\newcommand{\cC}{\mathcal{C}}
\newcommand{\cD}{\mathcal{D}}

\newcommand{\cL}{\mathcal{L}}
\newcommand{\cM}{\mathcal{M}}
\newcommand{\cN}{\mathcal{N}}

\newcommand{\cU}{\mathcal{U}}

\newcommand{\supp}{\textnormal{supp}}
\newcommand{\semc}{\mathcal{SC}}
\newcommand{\cDb}{\overline{\mathcal{D}}}
\newcommand{\gX}{\langle X\rangle}

\newcommand{\Sp}{\textnormal{Sp}}

\newcommand{\kz}{\ket{z}}

\newcommand{\Tr}{\textnormal{Tr}}

\newtheorem{question}{Question}
\newtheorem{theorem}{Theorem}
\newtheorem{defn}{Definition}
\newtheorem{cor}{Corollary}

\newtheorem{lemma}{Lemma}
\newtheorem{conj}{Conjecture}
\theoremstyle{remark}

\author{Nadish de Silva}
\author{Oscar Lautsch}

\affil{Department of Mathematics, Simon Fraser University, Burnaby, B.C., Canada}

\title{\vspace{-1.3cm}The Clifford hierarchy for one qubit or qudit}

\begin{document}

\date{} 
\maketitle
\vspace{0.3cm}
\begin{abstract}
The Clifford hierarchy is a nested sequence of sets of quantum gates that can be fault-tolerantly performed using gate teleportation within standard quantum error correction schemes.  The groups of Pauli and Clifford gates constitute the first and second `levels', respectively.  Non-Clifford gates from the third level or higher, such as the $T$ gate, are necessary for achieving fault-tolerant universal quantum computation.

Since it was defined twenty-five years ago by Gottesman-Chuang, two questions have been studied by numerous researchers.  First, precisely which gates constitute the Clifford hierarchy?  Second, which subset of the hierarchy gates admit efficient gate teleportation protocols?

We completely solve both questions in the practically-relevant case of the Clifford hierarchy for gates of one qubit or one qudit of prime dimension.  We express every such hierarchy gate uniquely as a product of three simple gates, yielding also a formula for the size of every level.  These results are a consequence of our finding that all such hierarchy gates can be expressed in a certain form that guarantees efficient gate teleportation.  Our decomposition of Clifford gates as a unique product of three elementary Clifford gates is of broad applicability.
\end{abstract}

\vspace{0.3cm}\tableofcontents

\newpage

\section{Introduction}
The challenges of manipulating quantum data in the presence of environmental noise necessitated the development of quantum error-correcting codes \cite{shor1996fault}.  Quantum gates can be performed directly on encoded data; doing so in a way that suppresses errors to a manageable level is the subject of quantum fault tolerance.  

The most commonly studied and utilised family of quantum error-correcting codes are stabiliser codes \cite{gottthesis}.  Members of the group of \emph{Clifford gates} are special in that they can typically be easily applied fault-tolerantly to data protected by stabiliser codes.  Quantum universality, however, further requires the ability to perform non-Clifford gates.  Gottesman-Chuang introduced the now-standard technique of \emph{gate teleportation} to fault-tolerantly implement certain non-Clifford gates using only Clifford gates supplemented with ancillary \emph{magic state} resources \cite{gottesman1999}.  In the same work, they introduced the \emph{Clifford hierarchy} (Definition \ref{ch}), a nested sequence of sets of gates that can be implemented in this way.

A significant practical barrier to achieving quantum universality using this method is the need to prepare magic states for every desired application of a non-Clifford gate.  The need to reduce the burden of this substantial resource overhead cost led to the study of more efficient protocols.  Hierarchy gates that are diagonal in the computational basis can be implemented using magic states that are half the size of the magic states required in the standard gate teleportation protocol \cite{zhou2000methodology}.  For example, the $T$ gate, which is a diagonal third-level gate, can be implemented using the magic state $(\I \otimes T) \ket{\Psi_{\text{Bell}}}$ using the standard protocol or using $T\ket{+}$ with the efficient one.  The efficient gate teleportation protocol for diagonal gates was generalised to also implement `nearly diagonal' \emph{semi-Clifford gates}, i.e.\ those gates $G$ such that  $G = C_1 D C_2$ for $C_1,C_2$ being Clifford gates and $D$ a diagonal hierarchy gate \cite{zeng2008semi}.  This decomposition of semi-Clifford gates is highly nonunique.

The important role that the Clifford hierarchy and semi-Clifford gates play in fault-tolerant quantum computation motivates two key questions.

\begin{question}
    Precisely which gates belong to the $k^{\text{th}}$ level of the Clifford hierarchy? 
\end{question}

\begin{question}
    For which $(d, n, k)$ are all $k^{\text{th}}$-level gates of $n$ qubits or qudits (of dimension $d$) semi-Clifford? 
\end{question}

\subsection{Prior work}

Research into the structure of the Clifford hierarchy and semi-Clifford gates has remained active from their discovery twenty-five years ago to the present \cite{anderson2024groups,meandimin,cui2017diagonal,oldpaper,gottesman1999,  he2024permutation, pllaha2020weyl,rengaswamy2019unifying,zeng2008semi,zhou2000methodology}.

The most progress towards answering the first question was made by Cui-Gottesman-Krishna \cite{cui2017diagonal}, who gave an explicit description of the diagonal hierarchy gates in terms of certain polynomials over finite fields.  Their results apply to both the case of qubit gates and, more generally, to qudit gates of prime dimension.

Zeng-Chen-Chuang \cite{zeng2008semi} gave a proof, aided by exhaustive computations, that all  hierarchy gates of one or two qubits are semi-Clifford.  They also showed that there exist three-qubit gates in the $k^{\text{th}}$ level that are \emph{not} semi-Clifford whenever $k \geq 4$.  Beigi-Shor \cite{beigi2009c3} reported a construction of Gottesman-Mochon showing that there exist $n$-qubit gates in the third level that are not semi-Clifford whenever $n \geq 7$.

%Recently, more attention has been paid to the higher-dimensional \emph{qudit} case.  Qudit-based computation promises increased capacity and efficiency \cite{quditmsd,tonchev2016quantum}.  The advantages of qudit-based computation has led to rapidly accelerating development by experimentalists \cite{chi2022programmable,chizzini2022molecular,karacsony2023efficient,low2020practical,ringbauer2022universal,seifert2022time,wang2018proof}.  %In the future, once quantum technology has progressed, we expect qudit-based computation to become commonplace. 

More recent work of the first author \cite{oldpaper} generalised the efficient gate teleportation protocol to higher prime dimensions and initiated the study of qudit semi-Clifford gates.  In this work, it was proved that all third-level gates of one qudit are semi-Clifford.  Further work of the first author and Chen utilised tools of algebraic geometry to establish that all third-level gates of two qudits are semi-Clifford \cite{meandimin}.  

The result of the present work is the first to allow both $d$ and $k$ to vary.

\subsection{Summary of main results}

In this work, we give a complete answer to Questions 1 and 2 in the case of one qubit or qudit of prime dimension.

\begin{itemize}
    \item We show that all one-qudit gates from any level of the Clifford hierarchy are semi-Clifford (Theorem \ref{thm: Proof of semi-Clifford conjecture}).
    \item To establish the above result, we show how to uniquely decompose a Clifford gate as a product of three elementary Clifford gates: $MDP$ (Lemma \ref{lemma: Normal Form for Clifford Gates}). Here $M$ is chosen from a finite set (see Definition \ref{defn: E and M}), $D$ is diagonal, and $P$ is a permutation gate.  This normal form for Clifford gates is of broader applicability.
    \item Combining the above two results yields a unique decomposition of any non-Clifford one-qubit or one-qudit gate from any level of the Clifford hierarchy as a product of three simple gates: $MDC$  (Theorem \ref{thm: Normal form for semi-Cliffords}).  Here, $M$ is as above, $D$ is a diagonal hierarchy gate, and $C$ is any Clifford gate; every $MDC$ is a hierarchy gate. 
    \item This yields a formula that counts the size of every level of the Clifford hierarchy (Corollary \ref{cor: size of the hierarchy}).
    
\end{itemize}

\section{Background}

Below, we review the mathematical definitions that we will require to state and prove our results. 

This work builds upon an algebraic framework for studying the Clifford hierarchy based on the discrete Stone-von Neumann theorem that was articulated in \cite{oldpaper}.  It contains a more detailed introduction to the basic concepts studied below; it explicitly treats the general multiqudit case.  Since, in this work, we are concerned with the one-qubit or one-qudit case, our presentation below is restricted accordingly for simplicity.  In this section, we assume that $d$ is an odd prime but all of the background material extends to qubits with minor modifications; we explicitly mention the qubit case modifications only in the instances where we will require them.

\subsection{Notational conventions}

Let $d \in \N$ be an odd prime.  Denote by $\omega = e^{i 2 \pi / d}$ the $d$-th primitive root of unity.  The \emph{computational basis} of $\C^d$ is the standard basis and is denoted by the kets $\kz$ for $z \in \Z_d$ where $\Z_d$ is the field of integers modulo $d$.  

The unitary group of $\C^d$ is denoted $\cU$.  When we refer to a unitary $U$ \emph{up to phase} we mean its equivalence class $[U]$ under the equivalence relation $U \sim V \iff U = e^{i\theta}V$ for some $\theta \in \R$.  

We denote by $\T$ the complex unit circle group.  For a function $f: \Z_d \to \C$ we denote by $D[f]$ the diagonal matrix with $f(z)$ as its $z^\text{th}$ diagonal entry: $D[f] \ket{z} = f(z) \ket{z}$.

\subsection{Pauli gates}

The Pauli gates form the basis of the stabiliser codes used for quantum error correction.  Quantum data are encoded as simultaneous eigenvectors of commuting sets of Pauli gates; projections onto eigenspaces of Pauli gates represent the elementary measurements of stabiliser theory.  

The \emph{basic Pauli gates} for a single qudit are $Z, X \in \cU$: $Z \ket{z} = \omega^{z} \ket{z}$ and $X \ket{z} = \ket{z +1}$ where the addition is taken modulo $d$.  These unitaries have order $d$.  They also satisfy the Weyl commutation relations: \begin{equation}Z X = \omega X Z.\end{equation}

\begin{defn}\label{paulidefn}The group of \emph{Pauli gates} on one qudit is the subgroup of $\cU$ generated by the basic Pauli gates: \begin{equation}\cC_1 = \{\omega^c Z^{p} X^{q}\ \; | \; c \in \Z_d, (p,q) \in \Z_d^{2} \}.\end{equation}
\end{defn}

\subsection{The Pauli basis}

Using the above commutation relations, it is straightforward to check that Pauli gates from different phase classes are orthogonal with respect to the Hilbert-Schmidt inner product.  For each choice of $p,q \in \Z_d$, one representative from each phase class is carefully chosen.

\begin{defn}\label{def: W(p,q)}
        The \emph{Pauli basis} for $M_d(\C)$ is a set of Pauli gates defined for each $p, q\in \Z_d$ as
    \begin{equation}W(p, q)=\omega^{-2^{-1}pq}Z^p X^q.\end{equation} 
\end{defn}

%\red{Qubit case?  Do we need to specify $d$ is odd yet? We do need $d$ to be odd here (because of the $2^{-1}$). We want the Pauli basis in the qubit case to be $\{\I, Z, X, iZX\}$, so the only real change is the phase of $W(1,1)$ being $i$. Everything in this section holds true for qubits with this choice of the Pauli basis.}

The pairs $(p,q) \in \Z_d^2$ admit a symplectic product calculated within $\Z_d$: 
\begin{equation}
\label{sympprod}
    [(p,q),(p',q')] = p \cdot q' - p' \cdot q.
\end{equation}  
The phases for the Pauli basis were chosen so that $W(p,q)^* = W(-p,-q)$ and to satisfy the commutation relation: 
\begin{equation}
\label{commpaulisymp}
    W(p,q) W(p',q') = \omega^{[(p,q),(p',q')]}   W(p',q') W(p,q).
\end{equation}

We will frequently decompose gates using the Pauli basis and collate the resulting coefficients.

\begin{defn}\label{def: Def of f_M}
    Let $M \in M_d(\C)$ be a matrix. We define the function $f_M:\Z_d^2\to \C$ by 
\begin{equation}f_M(p,q) = d^{-1} \, \Tr(W(-p,-q)M).\end{equation}
    \end{defn}

Cyclicity and linearity of trace straightforwardly yield a convolution formula for products of matrices.
     \begin{lemma}
    \label{lemma: Multiplying f_V functions}
    For any $U,V\in M_d(\C)$ and $p, q\in \Z_d$, we have 
    \begin{align}
        f_{UV}(p,q)&=\sum_{i,j\in \Z_d}\omega^{2^{-1}(qi-pj)}f_U(i,j)f_V(p-i, q-j)\\
        &=\sum_{i,j\in \Z_d}\omega^{-2^{-1}(qi-pj)}f_U(p-i,q-j)f_V(i, j).
    \end{align} 
\end{lemma}

We will be particularly concerned with the set of Pauli basis elements for which a given matrix has nonzero coefficients.  Pllaha et al. previously studied Clifford and third-level gates via their Pauli support in the qubit case \cite{pllaha2020weyl}.

\begin{defn}\label{def: pauli support}
    The \emph{Pauli support} of a matrix $M \in M_d(\C)$, denoted $\supp(f_M)$, is the set of points $(p,q)\in \Z_d^2$ for which $f_M(p,q)\neq 0$.
\end{defn}

    We can represent the Pauli support of $M\in M_d(\C)$ by creating a diagram for the \textit{phase plane} $\Z_d^2$ and highlighting the points where $f_M(p,q)$ is nonzero. An example with $d=5$ is shown below.
    \begin{figure}[H]
        \centering
        \includegraphics{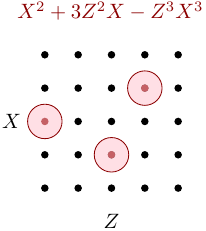}
        \caption[An Example of Pauli Supports]{The Pauli support of the linear combination $X^2+3Z^2X-Z^3X^3$. The powers of $Z$ range from $0$ to $d-1$ along the horizontal axis, while the powers of $X$ range along the vertical axis.}
        \label{fig:Pauli Support Example}
    \end{figure}

    \begin{lemma}
    A gate of order $d$ is a Pauli gate if and only if its Pauli support is of size $1$.
    \label{obs: Pauli condition on f_V functions}
\end{lemma}

\subsection{Conjugate pairs}

If we conjugate $Z,X \in \cU$ by any unitary gate $G$, the gates $G_Z = GZG^*$ and $G_X = GXG^*$ also satisfy the Weyl commutation relations: ${G_Z}^d = {G_X}^d = \I$ and $G_ZG_X = \omega G_XG_Z$.  Remarkably, if $(U,V)$ is any pair of unitaries satisfying these relations, there exists a gate $G$, unique up to phase, such that $U = G_Z$ and $V = G_X$.  

\begin{defn}
An ordered pair of unitaries $(U,V) \in \cU^2$ is a \emph{conjugate pair} if
\begin{enumerate}\itemsep0em 
\item $U^d = \I$ and $V^d = \I$,
\item $UV = \omega VU$.
\end{enumerate}
\end{defn}

As mentioned, conjugate pairs are in bijective correspondence with one-qudit gates up to phase \cite[Lemma 3.4]{oldpaper}.

\begin{lemma}[Discrete Stone-von Neumann theorem, single-qudit version]\label{uniqueunitary}
Suppose $(U, V) \in \cU^2$ is a conjugate pair.  There is a unitary $G$, unique up to phase, such that $U = GZG^*$ and $V = GXG^*$.
\end{lemma}
We call $(GZG^*, GXG^*)$ the pair \textit{corresponding} to the gate $G$.  Below, we study gates via their conjugate pairs.

\subsection{Clifford gates}

Clifford gates typically have a simple fault-tolerant implementation on stabiliser codes.  

\begin{defn}\label{cliffdefn}The group of \emph{Clifford gates} is the normaliser of the group of Pauli gates as a subgroup of $\cU$: \begin{equation}\cC_2 = \{C \in \cU \;|\; C \cC_1 C^* \subseteq \cC_1 \}.\end{equation}
\end{defn}

%\begin{lemma}\label{commcliff}
%For any Clifford gate $C \in \cC_2$ and Pauli gate $P_1 \in \cC_1$, there exists a Pauli gate $P_2 \in \cC_1$ such that \begin{equation}CP_1 = P_2C.\end{equation}
%\end{lemma}

\begin{defn}A gate $G_1 \in \cU$ is \emph{Clifford-conjugate} to $G_2 \in \cU$ if there exists a Clifford gate $C \in \cC_2$ such that $G_1 = CG_2C^*$.
\end{defn}

\begin{defn}\label{defn: sim cliff conj}A conjugate pair $(U,V) \in \cU^2$ is \emph{jointly Clifford-conjugate} to a conjugate pair $(U',V') \in \cU^2$ if there exists a Clifford gate $C \in \cC_2$ such that $U' = CUC^*$ and $V' = CVC^*$.
\end{defn}

\noindent

We define the group of Clifford gates up to phase: $[\cC_2] = \cC_2 / \mathbb{T}$.  Clifford gates of one qubit or qudit, up to phase, are in correspondence with affine symplectic transformations of $\Z_d^{2}$ \cite{gross2006hudson, pllaha2021clifford}. 

The group of symplectic linear transformations $\Sp(1,\Z_d)$ coincides with $\textnormal{SL}(2,\Z_d)$.  The group $\Sp(1,\Z_d) \ltimes \Z_d^{2}$ of affine symplectic transformations of $\Z_d^{2}$ are pairings of $2 \times 2$ symplectic matrices over $\Z_d$ and translations in $\Z_d^{2}$ with the composition law: \begin{equation}(A,v) \circ (B,w) = (AB, Aw + v).\end{equation}  There is a projective representation $\rho: \Sp(1,\Z_d) \ltimes \Z_d^{2} \to [\cC_2]$ that is an isomorphism between the groups of affine symplectic transformations and Clifford gates up to phase.

We will require an explicit representation of the symplectic group given by Neuhauser \cite{neuhauser2002explicit}; we have slightly changed convention by associating a pair $(p,q) \in \Z_d$ with $Z^p X^q$ rather than $X^{-p} Z^q$.  First, note that $\Sp(1,\Z_d)$ is generated by the set of matrices of the form   
    \begin{align}
    \qquad 
        \begin{pmatrix}
            a^{-1} & 0\\
            0 & a
        \end{pmatrix},
        \qquad \text{and}
        \qquad
        \begin{pmatrix}
            1 & b\\
            0 & 1
        \end{pmatrix},
        \qquad \text{and}
        \qquad
        \begin{pmatrix}
            0 & 1\\
            -1 & 0
        \end{pmatrix},
    \end{align}
    where $a\in \Z_d^*$ and $b\in \Z_d$.

\begin{theorem}[Neuhauser's representation of $\Sp(1, \Z_d)$]
    \label{Thm: Neuhauser's representation}
       There is a projective representation $\mu:\Sp(1, \Z_d)\to[\cC_2]$, defined by sending generators of $\Sp(1, \Z_d)$ to the following gates, up to phase:
       \begin{itemize}
       \item For all $a\in \Z_d^*$, $\mu\left(\begin{psmallmatrix}
            a^{-1} & 0\\
            0 & a
        \end{psmallmatrix}\right)$ is the Clifford permutation gate that sends $\ket{z}$ to $\ket{az}$.
           \item  For all $b\in \Z_d$, $\mu\left(\begin{psmallmatrix}
            1 & b\\
            0 & 1
        \end{psmallmatrix}\right)$ is the diagonal gate $D[f]$, where in the qudit case $f:\Z_d\to \C$ is the map $f(z)~=~\omega^{2^{-1}bz^2}$.
        In the qubit case, we instead have $f(z)=i^{bz}$, as noted in \cite{pllaha2021clifford}.
        \item $\mu\left(\begin{psmallmatrix}
            0 & 1\\
            -1 & 0
        \end{psmallmatrix}\right)$ is the \emph{Hadamard gate} $H$ with $H_{ij} = d^{-2^{-1}}\omega^{ij}$.
       \end{itemize}
    \end{theorem}

We can use $\mu$ to define the isomorphism $\rho: \Sp(1,\Z_d) \ltimes \Z_d^{2} \to [\cC_2]$ by:

\begin{equation}\rho((S, (p,q))) = [W(p,q)\mu(S)].\end{equation}

\begin{defn}
    A \emph{symplectic Clifford gate} $C \in \cC_2$ is one such that there exists $S \in \Sp(1, Z_d)$ such that $\mu(S) = [C]$.
\end{defn}

If $C$ is a symplectic Clifford gate and $S\in \Sp(1, \Z_d)$ is the symplectic transformation corresponding to $[C]$ under Neuhauser's representation, then for all $p,q\in \Z_d$,
    \begin{align}\label{eqn: conjugating a Pauli by a Clifford}
        CW(p,q)C^*=W(S(p,q)).
    \end{align}
%\red{In the qubit case, we only have $CW(p,q)C^*=\pm W(S(p,q))$, but this does not effect our proofs of the normal forms.}

The group of Clifford gates is a maximal nondense subgroup of the unitary group, and thus, approximately performing arbitrary computations requires the ability to fault-tolerantly perform a non-Clifford gate.

\subsection{The Clifford hierarchy}

The standard choices for non-Clifford gates to supplement the group of Clifford gates in order to achieve universal quantum computation come from the \emph{Clifford hierarchy} \cite{gottesman1999}.  This is a recursively-defined and nested sequence of subsets of $\cU$.  The groups of Pauli and Clifford gates form the first and second \emph{levels} respectively: $\cC_1 \subset \cC_2 \subset \cC_3 \subset ...$.  

Non-Clifford gates $G$ that are in the third level or higher can be implemented fault-tolerantly on encoded data indirectly via a gate teleportation protocol.  Such a protocol takes as input an arbitrary data state $\ket{\psi}$ and a resource magic state $(\I \otimes G) \ket{\Psi_{\text{Bell}}}$ and as output produces, using only gates from lower levels of the hierarchy and standard measurements, $G \ket{\psi}$: the data state with the desired gate applied to it.

\begin{defn}[Gottesman-Chuang, 1999]\label{ch}The \emph{Clifford hierarchy} is an inductively defined sequence of sets of gates.  For $k \geq 2$, the \emph{$k\textsuperscript{th}$ level of the Clifford hierarchy} is the set:  \begin{equation}\cC_k = \{G \in \cU \;|\; G \cC_1 G^* \subseteq \cC_{k-1}\}.\end{equation}
\end{defn}

While the first two levels form groups, higher levels do not.  The sets $\cC_k$ are closed under left or right multiplication by Clifford gates: for $k \geq 2$, $\cC_2 \, \cC_k \, \cC_2 = \cC_k$ \cite{zeng2008semi}. 

The following definition is used to characterise those conjugate pairs that correspond to gates of the Clifford hierarchy \cite[Theorem 3.11]{oldpaper}.

\begin{defn}\label{def: k-closed}
A conjugate pair $(U,V) \in \cU^2$ is  \emph{$k$-closed} if it generates a group of $k\textsuperscript{th}$-level gates.  Equivalently: \begin{equation}\{U^p V^q \;|\; (p,q) \in \Z_d\} \subseteq \cC_k.\end{equation} 
\end{defn}

\begin{theorem}\label{Ckfromtuples}
Gates of the $k+1\textsuperscript{th}$ level of the Clifford hierarchy, up to phase, are in bijective correspondence (via Lemma \ref{uniqueunitary}) with $k$-closed conjugate pairs.
\end{theorem}

\subsection{Diagonal hierarchy gates}

Of particular interest are the \emph{diagonal} Clifford hierarchy gates.  Every level restricted in this way forms a group.

\begin{defn}\label{diagkdefn}The group of \emph{phaseless diagonal $k\textsuperscript{th}$-level gates} is: \begin{equation}\cD_k = \{D \in \cC_k \;|\; D \text{ is a diagonal gate and } D\ket{0} = \ket{0}\}.\end{equation}

We will denote the group of \emph{diagonal $k\textsuperscript{th}$-level gates} by $\cDb_k = \T \, \cD_k$.

\end{defn}

The condition that $D$ fixes $\ket{0}$ ensures that $\cD_k$ contains precisely one representative from every phase class.  Cui-Gottesman-Krishna \cite{cui2017diagonal} characterised the diagonal hierarchy gates exactly for any number of qubits or qudits. We give a simplified restatement of their result restricted to the one-qudit case.  This appeared, along with a concise proof of its correctness, in earlier work of the first author \cite{oldpaper}.

First, note that any integer $k \geq 1$ can be uniquely expressed as $k = (m-1)(d-1) + a$ with $a \in \{1, ..., d-1\}$.  Given $m \geq 1$, denote the $d^m$-th primitive root of unity: $\omega_m = e^{i \,  2 \pi / d^m}$.  Denote by $\crr_k$ the set of \emph{rank-$k$ polynomials}: \begin{align}
\crr_k &= \Set{ \begin{array}{l} \phi: \Z_{d^m} \to \Z_{d^m} \smallskip\\  
\phi(z) = \sum_{j=1}^{d-1} \phi_j z^j \end{array} \ | \begin{array}{l}
    \phi \text{ has degree at most } d-1, \; \phi(0) = 0\smallskip\\
    \phi_{a+1}, ..., \phi_{d-1} \equiv 0 \;\;\mathrm{ (mod }\;d) \\
  \end{array}}.
\end{align}

\begin{theorem}\label{CGK}    
$\cD_k = \{D[{\omega_m}^\phi] \;|\; \phi \in \crr_k \}$ for all $k \geq 1$.
\end{theorem}

One of our main results below (Theorem \ref{thm: Normal form for semi-Cliffords}) extends this classification from $\cD_k$ to all of $\cC_k$.

\subsection{Semi-Clifford gates}

Zhou-Leung-Chuang \cite{zhou2000methodology} introduced a simplified gate teleportation protocol, based on Bennett-Gottesman's one-qubit teleportation, capable of fault-tolerantly implementing certain qubit Clifford hierarchy gates using half the ancillary resources required in the original Gottesman-Chuang protocol.  This class of gates includes the diagonal Clifford hierarchy gates.  Zeng-Chen-Chuang \cite{zeng2008semi} introduced the notion of semi-Clifford gates which are `nearly diagonal' in the sense of being within Clifford corrections of diagonal gates:

\begin{defn}\label{sc}A gate $G \in \cC_k$ is \emph{semi-Clifford} if $G = C_1 D C_2$ where $C_1,C_2 \in \cC_2$ and $D$ is diagonal.\end{defn}

\begin{defn}For $k \geq 1$ the $k\textsuperscript{th}$-level semi-Clifford gates are: \begin{equation}\semc_k = \{G \in \cC_k \;|\; G = C_1 D C_2 \text{ where } C_1,C_2 \in \cC_2, \, D \in \cD_k\}.\end{equation}
\end{defn}

The following gate teleportation protocol for implementing the qudit semi-Clifford gate $G = C_1 D C_2$ using the magic state $\ket{M} = D\ket{+}$ was introduced in \cite[\S5(a)]{oldpaper}.  It ensures that the notion of semi-Clifford gate is still relevant in the qudit setting.

    \begin{figure}[H]
        \centering
        \begin{center}
\begin{quantikz}
\ket{0}\gategroup[wires=1,steps=3, style={dashed}]{{MAGIC STATE: $\ket{M} = D\ket{+}$}}    \& \gate{H} \& \gate{D} \& \ctrl{1} \& \gate{C_1}   \& \gate{C_1 D X^* D^* {C_1}^*} \& \qw \& \; G\ket{\psi} \\
\ket{\psi} \& \gate{C_2} \& \gate{H^2} \& \targ{} \& \qw    \& \meter{}  \vcw{-1}
\end{quantikz}
\end{center}
        \caption[Efficient gate teleportation]{The compact gate teleportation protocol for semi-Clifford gates.}
        \label{fig:Efficient gate teleportation}
    \end{figure}
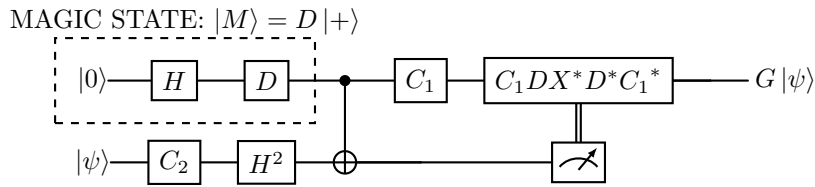

The following two lemmas are a direct consequence of Definition~\ref{sc}.

\begin{lemma}\label{cliffmultsemi}If $G \in \semc_k$ is semi-Clifford then, for any Clifford $C \in \cC_2$,  $GC$ and $CG$ are also semi-Clifford.
\end{lemma}

\begin{lemma}\label{cliffconjsemi}If $G_1 \in \cU$ is Clifford-conjugate to $G_2 \in \cU$, then $G_1$ is semi-Clifford if and only if $G_2$ is.
\end{lemma}

\section{Technical preliminaries}

In this section, we develop some tools and techniques necessary for our proof that are either novel or build substantially on prior work.  These have wider applicability beyond their use below.

\subsection{A normal form for Clifford gates}\label{Section: Normal From for Cliffords}
In this section, we develop a \textit{normal form} for $\cC_2$, which is a way to write each gate in $\cC_2$ uniquely as the product of three elementary Clifford gates.  For simplicity, we focus on the case of one qubit or qudit.  Our normal form extends to the multi-qubit or -qudit case; this is the subject of forthcoming work.

\begin{defn}
    A \emph{permutation gate} $P \in \cU$ is one defined by $P\ket{z} = \ket{\sigma(z)}$ for some permutation $\sigma: \Z_d \to \Z_d$.
\end{defn}

\begin{defn}
    The group of \emph{Clifford permutation gates}, denoted $\cN \subset \cU$, is the intersection of the group of Clifford gates with the group of permutation gates.
\end{defn}

Two examples include those arising from $\sigma(z) = az$ for any $a \in \Z_d^*$ (see Theorem \ref{Thm: Neuhauser's representation}) and $\sigma(z) = z + b$ for any $b \in \Z_d$ (which correspond to elements of $\gX$).  These, in fact, generate all examples.

\begin{lemma}
    Every Clifford permutation gate arises from $\sigma(z) = az + b$ for some $a \in \Z_d^*, b \in \Z_d$.
\end{lemma}

\begin{proof}
    Given $P \in \cU$ arising from $\sigma: \Z_d \to \Z_d$, $P Z P^* \ket{z} = \omega^{\sigma^{-1}(z)} \ket{z}$.  Thus $PZP^*$ is diagonal; it is a Pauli gate only if it is a power of $Z$ multiplied by a power of $\omega$.  Thus $\sigma$ is an affine linear map.
\end{proof}

\begin{lemma}
    $\gX$ is a normal subgroup of $\cN$: $\gX \trianglelefteq \cN$.
\end{lemma} 

We construct from $\cDb_k \leq \cU$ two other groups at each level of the hierarchy.

\begin{lemma}\label{lemma: DX and DN are groups}
    Let $k\in \N$, \begin{equation}\cDb_k\gX=\{DX^c \ | \ D\in \cDb_k, c\in \Z_d\}\end{equation}
    and, when $k \geq 2$,
    \begin{equation}\cDb_k\cN=\{DP \ | \ D\in \cDb_k, P\in \cN\}.\end{equation}
    Both of these subsets of $\cC_k$ are groups. Furthermore, for all applicable $k$, we have:
    \begin{align}
        \cDb_k\gX \trianglelefteq \cDb_{k+1}\gX, \qquad \cDb_k\gX\trianglelefteq \cDb_{k}\cN, \qquad \cDb_k\gX\trianglelefteq \cDb_{k+1}\cN.
    \end{align}
\end{lemma}
\begin{proof}
    That $\cDb_k\gX$, $\cDb_k\cN$ are subsets of $\cC_k$ follows from the closure of $\cC_k$ under multiplication by Clifford gates \cite{zeng2008semi}.  That they are groups, i.e. closed under inverses and products, follows straightforwardly by the fact that, for any $P \in \cN$, $PDP^* \in \cDb_k$ if and only if $D \in \cDb_k$.  

    Since the Clifford hierarchy is nested, $\cDb_k\leq \cDb_{k+1}$ for all $k$. Then as $\gX\trianglelefteq \cN$, both $\cDb_{k+1}\gX$ and $\cDb_{k}\cN$ are subgroups of $\cDb_{k+1}\cN$. Thus to show that $\cDb_k\gX$ is a normal subgroup of each of $\cDb_{k+1}\gX$, $\cDb_{k}\cN$, and $\cDb_{k+1}\cN$, it suffices to only show that $\cDb_k\gX$ is a normal subgroup of $\cDb_{k+1}\cN$. 
    
    Let $D_1X^b\in \cDb_k\gX$ and $D_2P\in \cDb_{k+1}\cN$. Then 
    \begin{align}
        (D_2P)(D_1X^b)(D_2P)^*&=D_2(PD_1P^*)(PX^bP^*)D_2^* \\
        &=D_2(PD_1P^*)X^cD_2^*\\
        &=(PD_1P^*)(D_2X^cD_2^*)\\
        &=(PD_1P^*)D'X^c\in \cDb_k\gX.
    \end{align}
    The second equality follows from $\gX \trianglelefteq \cN$.  The third equality follows from the fact that the diagonal gates $D_2$ and $PD_1P^*$ commute.  The final equality follows by noting that $D_2X^cD_2^*=D'X^c$ for some diagonal $D' \in \cDb_k$ and that $P \in \cC_2$ implies that $PD_1P^* \in \cDb_{k}$.
\end{proof}

By choosing convenient coset representatives for the subgroup $\cDb_2\cN$ of $\cC_2$, we construct a normal form for Clifford gates.

\begin{defn}\label{defn: E and M}In the qubit case, let $E \in \cU$ be the phase gate $S$; in the qudit case, let $E=D[h]$, where $h : \Z_d\to \C$ is the map $h(z) = \omega^{(z^2)}$.

We define $\cM \subset \cU$ to be the set of $d+1$ gates $\cM = \{\I\}\cup\{E^cH \ | \ c\in \Z_d\}$.
\end{defn}
 Note that by Theorem \ref{CGK}, $E$ is a generator for the group of diagonal symplectic Clifford gates. 
\begin{lemma}[Normal form for Clifford gates]
    The following is a normal form for Clifford gates:
    \begin{align}
        \cC_2=\cM \cDb_2\cN.
    \end{align}
    That is, every Clifford gate may be written uniquely, as $MDP$, where $M\in \cM$, $D\in \cDb_2$, and $P\in \cN$.
    \label{lemma: Normal Form for Clifford Gates}
\end{lemma}
\begin{proof}

    We begin by giving a normal form for $\Sp(1, \Z_d)$. By Theorem 3.1 of \cite{neuhauser2002explicit}, $\begin{psmallmatrix}
        s_{11} & s_{12}\\ s_{21} & s_{22}
    \end{psmallmatrix}\in\Sp(1, \Z_d)$ if and only if $s_{11}s_{22}-s_{12}s_{21}=1$.  Let $\begin{psmallmatrix}
        s_{11} & s_{12}\\ s_{21} & s_{22}
    \end{psmallmatrix}\in\Sp(1, \Z_d)$. This matrix may be written as 
    \begin{align}
        \begin{pmatrix}
        s_{11} & s_{12}\\ s_{21} & s_{22}
    \end{pmatrix}=\begin{pmatrix}
            1 & e \\0  &1
        \end{pmatrix}\begin{pmatrix}
            f& 0 \\0 &f^{-1}
        \end{pmatrix}
    \end{align}
    if and only if  $s_{21}=0$, $e=s_{11}s_{12}$, and $f=s_{11}$. On the other hand, this matrix may be written as 
    \begin{align}
        \begin{pmatrix}
        s_{11} & s_{12}\\ s_{21} & s_{22}
    \end{pmatrix}=        
            \begin{pmatrix}
                1 & g \\ 0 & 1
            \end{pmatrix}\begin{pmatrix}
                0 & 1\\ -1 & 0
            \end{pmatrix}\begin{pmatrix}
            1 & e \\ 0 &1
        \end{pmatrix}\begin{pmatrix}
            f& 0 \\0 &f^{-1}
        \end{pmatrix}
    \end{align}
    if and only if $s_{21}\neq 0$, $e=s_{21}s_{22}$, $f=-s_{21}$, and $g=s_{11}s_{21}^{-1}$. Letting  $A\subseteq\Sp(1, \Z_d)$ denote the set of matrices of the form $\begin{psmallmatrix}
        f & 0 \\ 0 & f^{-1}
    \end{psmallmatrix}$ and $B\subseteq\Sp(1, \Z_d)$ denote the set of matrices of the form $\begin{psmallmatrix}
        1 & e\\ 0 & 1
    \end{psmallmatrix}$, it follows that
    \begin{align}
        \left(\{\I\}\cup B\begin{psmallmatrix}
            0 &1 \\ -1 & 0
        \end{psmallmatrix}\right)BA
    \end{align}
    is a normal form for $\Sp(1, \Z_d)$. By Theorem \ref{Thm: Neuhauser's representation}, under the Neuhauser representation, elements of $A$ correspond to symplectic Clifford permutation gates, elements of $B$ correspond to diagonal symplectic Clifford gates, and $\begin{psmallmatrix}
        0 & 1 \\ -1&0
    \end{psmallmatrix}$ corresponds to $H$. 
    
    Thus, any symplectic Clifford gate may be written uniquely in the form $MD_1P$, where $D_1$ is a diagonal symplectic Clifford gate, $P$ is a symplectic Clifford permutation gate, and $M=\I$ or $M=D_2H$ for a diagonal symplectic Clifford gate $D_2$. Since any Clifford gate may be written uniquely as a Pauli gate multiplied by a symplectic Clifford gate, it follows that any Clifford gate may be written uniquely up to phase as
    \begin{align}
        MD_1PZ^pX^q
    \end{align}
    for some $p,q\in \Z_d$. Note that $D_1PZ^pX^q\in \cDb_2\cN$, so it may be written uniquely as $D'P'$ for some $D\in \cDb_2$ and $P\in \cN$.
    Hence, any Clifford gate may be written uniquely up to phase as $MD'P'$, where $D'\in \cD_2$, $P'\in \cN$, and $M=\I$ or $M=D_2H$ for a symplectic diagonal Clifford gate $D_2$. As $E$ is a generator for the group of symplectic diagonal Clifford gates, the claimed normal form follows.
\end{proof}
An identical argument shows that $\cC_2=\cDb_2\cN \cM{^{-1}}$ is also a normal form for Clifford gates, and we will use both of these normal forms in the remainder of this work.

\subsection{Almost diagonal, simplified, and simplifiable gates}

Here, we generalise two definitions made in the two-qudit third-level case in \cite{meandimin} to higher levels; while these definitions are easily extended to more qudits, we restrict them to one qudit for simplicity.

\begin{defn}
    A hierarchy gate $G \in \cC_k$ is \emph{almost diagonal} if $G \in \cDb_k \gX$. 
\end{defn}

\begin{defn}
A hierarchy gate $G \in \cC_k$ is \emph{simplified} if its conjugate pair contains only almost diagonal gates.
\end{defn}

We will show below in Lemma \ref{lemma: DC iff simplified} that a gate is simplified if and only if it is of the form $DC$ with $D$ a diagonal hierarchy gate and $C$ a Clifford gate.

\begin{defn}
A hierarchy gate $G \in \cC_k$ is \emph{simplifiable} if it is Clifford-conjugate to a simplified gate.
\end{defn}

It will immediately follow from Lemma \ref{lemma: DC iff simplified} that a gate is semi-Clifford if and only if it is simplifiable.  This reduces our primary goal to showing that every hierarchy gate is simplifiable.

That every third-level $n$-qudit gate is simplifiable \cite[Lemma 5.7]{oldpaper} was previously shown using a classification of the maximal abelian subgroups of the symplectic group over finite fields due to Barry \cite{barry1979large}.  In Section \ref{Section: Conjugate Pairs Corresponding to Semi-Cliffords}, we will show this to be true of a one-qudit gate from any level.

\subsection{Checking Clifford-conjugate equivalence}

In the qudit case, the Pauli basis decomposition and projective representation of the Clifford group give a convenient characterisation of when two matrices are Clifford-conjugate.  This characterisation does not seem to appear explicitly in the literature, so we state it here.  It extends naturally to higher numbers of qudits.  We omit the proofs as they are relatively straightforward.

    \begin{lemma}
    Let $M, M'\in M_d(\C)$. Then $M$ is Clifford-conjugate to $M'$ if and only if there exists $(S, (p', q'))\in \Sp(1, \Z_d)\ltimes \Z_d^2$ such that, for all $p,q\in \Z_d$,
    \begin{align}
        f_{M}(p,q) = w^{[(p, q),(p',q')]} f_{M'}(S(p,q)).
    \end{align}
    \label{lemma: Clifford-conjugate correspondence in f_M function}
\end{lemma}
\begin{lemma}If $M, M' \in M_d(\C)$ are Clifford-conjugate, then there exists $S\in \Sp(1, \Z_d)$ that restricts to a bijection between the Pauli support of $M$ and that of $M'$.
\end{lemma}

\section{Main Results}\label{mainresult}

Our main technical goal is to prove, in the one-qudit case, that all hierarchy gates are semi-Clifford.  This was already shown in the one-qubit case by Zeng-Chen-Chuang \cite{zeng2008semi}.  We combine these results with the normal form for Clifford gates established above in both qubit and qudit cases, to find normal forms for the Clifford hierarchy.

\subsection{Overview}

\begin{enumerate}
    \item In Section \ref{Section: Conjugate Pairs Corresponding to Semi-Cliffords}, we establish an equivalence between a gate being semi-Clifford and simplifiable.

    \item In Section \ref{Section: semi-cliff via pauli}, we find that a one-qudit gate is semi-Clifford if and only if  the Pauli supports of its conjugate pair lie within parallel lines in the phase plane. 

    \item In Section \ref{Section: Proof of the Conjecture}, we inductively prove that every one-qudit hierarchy gate is semi-Clifford via Pauli supports.

    \item In Section \ref{Section: Structure of the hierarchy} we use  that all hierarchy gates are semi-Clifford, and our normal form for Clifford gates, to explicitly describe the Clifford hierarchy of one qubit or qudit and find a formula for the size of every level.
\end{enumerate}

\subsection{Semi-Clifford gates and simplifiability}\label{Section: Conjugate Pairs Corresponding to Semi-Cliffords}

\begin{lemma}\label{lemma: SC implies simplifiable}
    Every semi-Clifford gate is simplifiable.
\end{lemma}

\begin{proof}
    
Every Pauli gate is simplified.  Thus, let $k \geq 1$ and $G = C_1DC_2\in \cC_{k+1}$ be semi-Clifford, where $C_1, C_2\in \cC_2$ and $D\in \cD_{k+1}$.  Then $C_1^* G C_1 = D C$ with $C = C_2 C_1$, and Lemma \ref{lemma: DX and DN are groups} ensures that both elements of the pair corresponding to $DC$ are in $\cDb_{k-1}\gX$, so $DC$ is simplified.
\end{proof}

In this section, we show that the converse of this also holds: that every simplifiable gate is semi-Clifford.  The following sequence of lemmas allows us to precisely characterise simplified gates as those of the form $DC$; a gate is semi-Clifford if and only if is Clifford-conjugate to a gate of the form $DC$.

\begin{lemma}
    Suppose $(U, V_1)$ and $(U, V_2)$ are both conjugate pairs. Then $U$ commutes with $V_1V_2^*$.
\end{lemma}

This observation allows us to show that in a conjugate pair of the form $(D_1X^{q_1}, D_2X^{q_2})$ with $D_1$ and $D_2$ diagonal, $D_2$ is almost completely determined by $D_1$.

\begin{lemma}
        Suppose that $(D_1X^{q_1}, D_2X^{q_2})$ and $(D_1X^{q_1}, D_3X^{q_2})$ are conjugate pairs, where $D_1, D_2$, and  $D_3$ are diagonal gates and ${q_1}\neq 0$. Then $D_3=\omega^c D_2$ for some $c\in \Z_d$.
    \label{cor: Diagonals uniquely determined by commutation relations}
\end{lemma}
\begin{proof}
    By the previous lemma, $D_1X^{q_1}$ must commute with $(D_2X^{q_2})(D_3X^{q_2})^*=D_2D_3^*$.   Then $X^{q_1}$ must commute with $D_2D_3^*$. As ${q_1}\neq 0$ and $d$ is prime, this implies that $D_2D_3^*$ is a constant multiple of $\I$. Since both $D_2X^{q_2}$ and $D_3X^{q_2}$ have order $d$, this constant multiple must be a power of $\omega$.
\end{proof}

The next lemma is used to construct the semi-Clifford gate corresponding to a conjugate pair of gates in $\cDb_k\gX$. 

\begin{lemma}
    Let $D'\in \cDb_k$ be a diagonal gate with $\det(D')=1$. Then for any fixed $c\in \Z_d^*$, there exists a diagonal gate $D\in \cD_{k+1}$ such that $D'=DX^{c}D^*X^{-c}$.
    \label{lemma: Any D_1 is DX^cD^*X^-c}
\end{lemma}
\begin{proof}
    Fix $c\neq 0$ in $\Z_d$. As $D'$ is a unitary diagonal, we have $D'=D[f]$ for some $f:\Z_d\to \mathbb{T}$. For any function $g(z):\Z_d\to \C$, if $D=D[g(z)]$, then $X^cD^*X^{-c}=D[g(z-c)^*]$, and so $DX^cD^*X^{-c}=D[g(z)g(z-c)^*]$. Hence, $D'=DX^{c}D^*X^{-c}$ if and only if 
    \begin{align}
       f(z)=g(z)g(z-c)^* 
       \label{eq: First eq in Any D_1 is DX^cD^*X^-c}
    \end{align}
    for all $z\in \Z_d$. We will show that a function $g$ satisfying \eqref{eq: First eq in Any D_1 is DX^cD^*X^-c} exists. Note that as $c$ is nonzero, we may replace $z$ with $cz$ in  \eqref{eq: First eq in Any D_1 is DX^cD^*X^-c} to find the equivalent equation
    \begin{align}
        f(cz)=g(cz)g(c(z-1))^*.
        \label{eq: Second eq in Any D_1 is DX^cD^*X^-c}
    \end{align}
    We construct such a $g$ by defining 
    \begin{align}
        g(cz)=\prod_{j=0}^{z}f(cj)
    \end{align}
    for all $z \in \Z_d$. This ensures that  \eqref{eq: Second eq in Any D_1 is DX^cD^*X^-c} is satisfied for all $1\leq z\leq d-1$. Furthermore, since $\det(D')=1$ the product of all the diagonal entries of $D'$ is 1. We thus have $\prod_{j=0}^{d-1}f(cz)=1$, and so $g(c(d-1))=1$. Then $g(0)g(c(d-1))^*=f(0)$, so  \eqref{eq: Second eq in Any D_1 is DX^cD^*X^-c} also holds for $z=0$. Since $f(z)\in \mathbb{T}$ for all $z\in \Z_d$, we have $g(z)\in \mathbb{T}$ for all $z$ as well, so $D=D[g]$ is unitary.
    Thus $D$ is a diagonal gate such that $D'=DX^cD^*X^{-c}$.

    To see that this $D \in \cDb_{k+1}$, note that $DX^cD^*=D'X^c\in \cDb_k\gX$. As $D$ and $Z$ commute, $D$ conjugates both $Z$ and $X^c$ into $\cDb_k\gX$. Since $Z$ and $X^c$ generate the Pauli group and $\cDb_k\gX$ forms a group, it follows that $D$ conjugates every Pauli into $\cDb_k\gX$, and so that $D\in\cDb_{k+1}$.  We choose the phase of $D$ so that $D\in\cD_{k+1}$.
\end{proof}
    The requirement that $D'$ have determinant 1 in this lemma may appear to restrict its utility. However, we will only need to apply this lemma to gates $D'$ such that $D'Z^pX^q$ appears in a conjugate pair. As we observe below, in this scenario, requiring that $D'$ have determinant 1 poses no issues.  
\begin{lemma}
    Suppose that $D'Z^pX^q$ is conjugate to $Z$, where $D'$ is a diagonal gate. Then $\det(D')=1$.
\end{lemma}  
This follows from the fact that $\det(Z)=\det(X)=1$. We can now precisely characterise simplified gates.
\begin{lemma}\label{lemma: DC iff simplified}
    A gate $G \in \cC_{k}$ is simplified if and only if it is of the form $DC$, with $D\in \cD_{k}$ and $C\in \cC_2$.
\end{lemma}
\begin{proof}The case of Pauli gates is immediate; thus let $k \geq 2$.

    The proof of Lemma \ref{lemma: SC implies simplifiable} establishes that $G=DC$ is simplified. Conversely, suppose that $G$ is simplified.  Then
    \begin{equation}(U, V) = (G Z G^*, G X G^*) = (D_1X^{q_1}, D_2X^{q_2})\end{equation}
    is a conjugate pair with $U, V\in \cDb_{k-1}\gX$.
    Since $U, V$ do not commute, $q_1$ and $q_2$ cannot both be 0. Assume without loss of generality that $q_1\neq 0$; the proof of the case with $q_2\neq 0$ is completely analogous. There exist $p_1, p_2\in \Z_d$ such that $p_1q_2-p_2q_1=1$, and there exist $D_1', D_2'\in \cDb_{k-1}$ such that $D_1=D_1'Z^{p_1}$ and $D_2=D_2'Z^{p_2}$. Hence, \begin{equation}(U, V)=(D_1'Z^{p_1}X^{q_1}, D_2'Z^{p_2}X^{q_2}).\end{equation} Combining the preceding two lemmas, we have $D_1'=DX^{q_1}D^*X^{-q_1}$ for some diagonal gate $D\in\cD_{k}$. Thus
    \begin{align}
        U=D_1'Z^{p_1}X^{q_1}=DX^{q_1}D^*X^{-q_1}Z^{p_1}X^{q_1}=DZ^{p_1}X^{q_1}D^*.
    \end{align}
    As $p_1q_2-p_2q_1=1$, $(Z^{p_1}X^{q_1}, Z^{p_2}X^{q_2})$ is a conjugate pair, so \begin{equation}(DZ^{p_1}X^{q_1}D^*, DZ^{p_2}X^{q_2}D^*)=(U, DX^{q_2}D^*X^{-q_2}Z^{p_2}X^{q_2})\end{equation}
    is as well.
    By Lemma \ref{cor: Diagonals uniquely determined by commutation relations}, it follows that $D_2'=\omega^cDX^{q_2}D^*X^{-q_2}$ for some $c \in \Z_d$. Hence
    \begin{align}
        (U, V)=(DZ^{p_1}X^{q_1}D^*, D\omega^cZ^{p_2}X^{q_2}D^*).
    \end{align}
    As $(Z^{p_1}X^{q_1}, \omega^cZ^{p_2}X^{q_2})$ is a conjugate pair of Paulis, there exists a Clifford $C$ such that \begin{equation}(CZC^*, CXC^*)=(Z^{p_1}X^{q_1}, \omega^cZ^{p_2}X^{q_2}).\end{equation} Then $(U, V)=(DCZC^*D^*, DCXC^*D^*)$, so $DC$ is the gate corresponding to the pair $(U, V)$.
\end{proof}

Our desired equivalence follows immediately.
\begin{theorem}[Equivalence of semi-Cliffordness and simplifiability]
     A gate $G \in \cU$ is semi-Clifford if and only if it is simplifiable. 
    \label{cor: Pairs (CD_1X^aC^*, CD_2X^bC^*) correspond to semi-Cliffords}
\end{theorem}
\begin{proof}
    We established one direction of this equivalence in Lemma \ref{lemma: SC implies simplifiable}. For the converse, suppose that $G = C_1 G' C_1^*$ where $G' \in \cC_k$ is simplified.  Then $G = C_1 D C C_1^*$ with $D \in \cD_k$, $C \in \cC_2$ and is thus semi-Clifford with $C_2 = C C_1^*$.
\end{proof}

\subsection{Characterising semi-Cliffordness via Pauli supports}\label{Section: semi-cliff via pauli}

Here, we study the Pauli supports of the gates in conjugate pairs $(U, V)$ arising from semi-Clifford gates. Recall, for any $M\in M_d(\C)$, the definition of $f_M$ and the Pauli support of $M$ from Definitions \ref{def: Def of f_M} and \ref{def: pauli support}. 

We begin by understanding the Pauli supports of gates in $\cDb_k\gX$. By a \textit{line} in $\Z_d^2$, we refer to a 1-dimensional affine subspace of the phase plane $\Z_d^2$. Two lines are \textit{parallel} if they are affine translates of the same 1-dimensional linear subspace of $\Z_d^2$. A line is \textit{horizontal} if it is of the form $\{(p,q_0) \ | \ p \in \Z_d\}\subseteq \Z_d^2$ for some $q_0\in \Z_d$.  

%Note that the first coordinate  $p$ of a vector  $(p,q)\in \Z_d^2$ should be thought of as both the horizontal coordinate of the vector in the plane $\Z_d^2$, and as corresponding to the power of $Z$ in the Pauli gate $W(p,q)$. Similarly, the second coordinate $q$ should be thought of as both the vertical coordinate of the vector and as corresponding to the power of $X$ in $W(p,q)$. This is in keeping with the conventions shown in Figure \ref{fig:Pauli Support Example}, where the powers of $Z$ vary along the horizontal axis while the powers of $X$ vary along the vertical axis.  

\begin{lemma}\label{lemma: Pauli supports of DX}
    Let $G \in \cU$.  Then $G = DX^{q_0}$ for some diagonal gate $D\in \cU$ and $q_0 \in \Z_d$ if and only if the Pauli support of $G$ is contained in a horizontal line $\{(p,q_0) \ | \ p\in \Z_d\}\subseteq \Z_d^2$ for some $q_0\in \Z_d$.
\end{lemma}
\begin{proof}
    If $G=DX^{q_0}$ for some diagonal gate $D \in \cDb_k$ and $q_0\in \Z_d$, then for all $p,q\in\Z_d$, 
    \begin{equation}GW(-p, -q)=DX^{q_0}\omega^{-2^{-1}pq}Z^{-p}X^{-q}=D'X^{q_0-q}\end{equation}
    for some diagonal $D'$. If $q\neq q_0$, then the matrix $D'X^{q_0-q}$ has all zeros along its diagonal, and so has trace 0. Hence, $      f_G(p,q)=d^{-1}\Tr(GW(-p, -q))$
    is 0 as long as $q\neq q_0$. The Pauli support of $G$ is then contained in $\{(p,q_0) \ |\ p\in \Z_d\}$.

    The converse follows by noting that if $\supp(f_G) \subset \{(p,q_0) \ |\ p\in \Z_d\}$, then for some constants $c_p\in \C$,
    \begin{align}
        G=\sum_{p\in \Z_d}c_pZ^{p}X^{q_0}=\left(\sum_{p\in \Z_d}c_pZ^{p}\right)X^{q_0}.
    \end{align}\end{proof}
The Pauli supports of gates in $\cDb_k\gX$ can be visualised in Figure \ref{fig: Almost diagonal gates in the Pauli basis}.

\begin{figure}[H]
    \centering
    \includegraphics{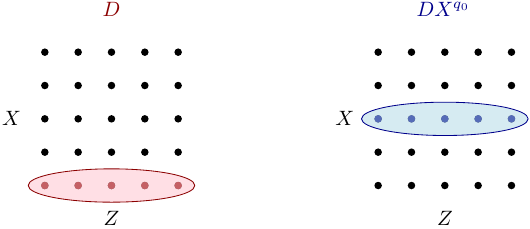}
    \caption[Gates in $\cDb_k\gX$ in the Pauli basis]{The Pauli support of a diagonal gate is contained in the linear subspace $\{(p, 0) \ | \ p\in \Z_d\}$, and in general, the Pauli support of $DX^{q_0}$ is contained in $\{(p,q_0) \ | \ p\in \Z_d\}$ if $D$ is diagonal.}
    \label{fig: Almost diagonal gates in the Pauli basis}
\end{figure}

There is an action of $\Sp(1, \Z_d)$ on the set of lines in $\Z_d^2$, with $S\in \Sp(1, \Z_d)$ sending $L$ to $S(L)$.   This induces an action of $\cC_2$ on these lines: for any $C\in \cC_2$, we may write $[C]$ uniquely as $[W(p,q)C']$ for some symplectic Clifford $C'$.  Letting the symplectic transformation corresponding to $C'$ be $S$, $C$ acts on a line $L\subseteq\Z_d^2$ by sending $L$ to $S(L)$.

Lemma \ref{lemma: Pauli supports of DX} means that horizontal lines will be especially important for us. We thus wish to understand which Clifford gates $C$ send horizontal lines to horizontal lines under the above action. This is equivalent to $C$ mapping the linear subspace $L_Z=\{(p, 0)\ | \ p\in \Z_d\}$ to itself under this action, since all horizontal lines are affine translates of $L_Z$. Finally, as $L_Z$ is spanned by $(1,0)$, this is equivalent to $C$ sending $(1, 0)$ to $(p, 0)$ for some $p\in \Z_d^*$. The Clifford gate $C$ does this if and only if the corresponding symplectic transformation $S$ is of the form
\begin{align}
    S=\begin{pmatrix}
        s_{11} & s_{12}\\
        0 & s_{22}
    \end{pmatrix}.
\end{align}
As shown in the proof of our Clifford normal form (Lemma \ref{lemma: Normal Form for Clifford Gates}), $S$ has the above form if and only if $C\in \cDb_2\cN$. We have thus proven the following lemma.
\begin{lemma}\label{lemma: DN stabilizes horizontal lines}
    Let $C\in \cC_2$. Then $C$ maps horizontal lines to horizontal lines under the above action if and only if $C\in \cDb_2\cN$.
\end{lemma}

We also wish to have a standard set of Clifford gates to take $L_Z$ to each 1-dimensional linear subspace of $\Z_d$. This is done in the following lemma, in which we use the set $\cM$ of $d+1$ gates from Definition \ref{defn: E and M}.

\begin{lemma}\label{lemma: M's correspond to slopes}
    Under the above action, each $M\in \cM$ maps $L_Z=\{(p, 0)\ | \ p\in \Z_d\}$ to a 1-dimensional linear subspace of $\Z_d^2$. Furthermore, for any $1$-dimensional subspace $L$, there is a unique  $M\in \cM$ that carries $L_Z$ to $L$.
\end{lemma}
\begin{proof}
    It suffices to show that each such $M$ maps $(1,0)$ into a distinct 1-dimensional subspace. The identity $\I$ maps $(1, 0)$ into $L_Z$, and by Theorem \ref{Thm: Neuhauser's representation}, the Clifford gate $E^cH$ corresponds to the symplectic matrix
    \begin{align}
        S=\begin{pmatrix}
            1 & 2c \\
            0 & 1
        \end{pmatrix}
        \begin{pmatrix}
            0 & 1 \\
            -1 & 0
        \end{pmatrix}
        =\begin{pmatrix}
            -2c & 1 \\
            -1 & 0
        \end{pmatrix}.
    \end{align}
    $S$ maps $(1, 0)$ to $(-2c, -1)$, which never lies in $L_Z$ and lies in a distinct 1-dimensional subspace for distinct values of $c$. Furthermore, any 1-dimensional subspace of $\Z_d^2$ is spanned by either $(1,0)$ or $(-2c, -1)$ for some $c\in \Z_d$.
\end{proof}
This also shows that $\Sp(1, \Z_d)$ acts transitively on the 1-dimensional subspaces of $\Z_d^2$. We can now understand exactly when a conjugate pair $(U, V)$ arises from a semi-Clifford gate in terms of the Pauli supports of $U$ and $V$.

\begin{theorem}\label{cor: Pauli support of conj. pairs}
    Let $G\in \cC_{k}$ and $(U, V)$ be the pair corresponding to $G$. Then $G$ is semi-Clifford if and only if the Pauli supports of $U$ and $V$ are contained in two parallel lines $L_U$ and $L_V$ in $\Z_d^2$. 
\end{theorem}
\begin{proof}
    By Theorem \ref{cor: Pairs (CD_1X^aC^*, CD_2X^bC^*) correspond to semi-Cliffords}, $G$ is semi-Clifford if and only if 
    \begin{equation}(U, V)=(CU'C^*, CV'C^*),\end{equation}
     for some $C\in\cC_2$ and $U', V'\in \cDb_k\gX$. By the Lemmas \ref{lemma: Pauli supports of DX} and \ref{lemma: Clifford-conjugate correspondence in f_M function}, this is equivalent to the existence of two horizontal lines $L_{U'}$ and $L_{V'}$ (containing the Pauli supports of $U'$ and $V'$), and the existence of a symplectic transformation $S\in \Sp(1, \Z_d)$ (corresponding to the symplectic part of $C$), such that the Pauli supports of $U$ and $V$ are contained in $S(L_{U'})$ and $S(L_{V'})$. If such $L_{U'}$, $L_{V'}$, and $S$ exist, then $L_U=S(L_{U'})$ and $L_V=S(L_{V'})$ are parallel since $S$ is linear and $L_{U'}$ and $L_{V'}$ are parallel. Conversely, if the Pauli supports of $U$ and $V$ are contained in parallel lines $L_U$ and $L_V$, then there exists $S\in \Sp(1, \Z_d)$ such that $L_{U'}=S^{-1}(L_U)$ and $L_{V'}=S^{-1}(L_V)$ are horizontal, since $\Sp(1, \Z_d)$ acts transitively on the 1-dimensional linear subspaces of $\Z_d^2$ by the above lemma.
\end{proof}
The proof of this corollary can be visualised with Figure \ref{fig:Conjugate pairs in the Pauli basis}.
\begin{figure}[H]
    \centering
    \includegraphics[page=3]{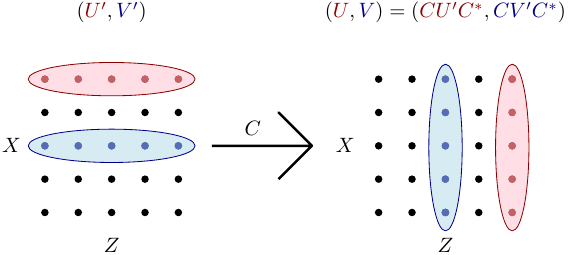}
    \caption[Proof of Theorem \ref{cor: Pauli support of conj. pairs}]{Proof of Theorem \ref{cor: Pauli support of conj. pairs}. The horizontal lines in the left diagram contain the Pauli supports of $U'$ and $V'$, which are mapped to the parallel lines on the right side of the diagram by $C$.}
    \label{fig:Conjugate pairs in the Pauli basis}
\end{figure}

\begin{comment}
    Let $C_1 D C_2$ be a semi-Clifford gate, where $D\in \cD_{k+1}$ and $C_1, C_2\in \cC_2$, and consider the $k$-closed conjugate pair corresponding to $C_1DC_2$. Both $C_2ZC_2^*$ and $C_1XC_1^*$ are Pauli, so both $DC_2ZC_2^*D^*$ and $DC_2XC_2^*D^*$ are in $\cC_k\gX$. It follows that the pair corresponding to 
\end{comment}

\subsection{Every hierarchy gate is semi-Clifford}\label{Section: Proof of the Conjecture}
Recall from Definition \ref{def: k-closed} that a pair $(U, V)$ is $k$-closed if the group $\langle U, V\rangle$ generated by $U$ and $V$ is contained in $\cC_k$. In this section, we will prove our main lemmas (\ref{lemma: Only gates in k-closed pairs are in D_k X} and \ref{lemma: Pair individually conjugate to DX is jointly conjugate}), which will show that every $k$-closed conjugate pair is jointly Clifford-conjugate (recall Definition \ref{defn: sim cliff conj}) to a pair of gates in $\cDb_k\gX$. Since gates at level $k+1$ correspond to $k$-closed pairs by Theorem \ref{Ckfromtuples}, the bijective correspondence established in Theorem \ref{cor: Pairs (CD_1X^aC^*, CD_2X^bC^*) correspond to semi-Cliffords} will then ensure that every gate at level $k+1$ is semi-Clifford. Using induction, we will then be able to show that for every odd prime $d$, every gate at every level of the one-qudit Clifford hierarchy is semi-Clifford.

We begin by proving that in two special cases, certain order-$d$ gates are Clifford-conjugate to gates in $\cDb_k\gX$.
\begin{lemma}
    Any Clifford gate $C \in \cC_2$ of order $d$ is Clifford-conjugate to a gate in $\cDb_2\gX$.\label{lemma: DX^c is Sylow d-subgroup}
\end{lemma}
\begin{proof}
    If $C$ appears in a conjugate pair, it has order $d$, so $[C]$ has order $d$ as well. Thus $[C]$ is in a Sylow $d$-subgroup of $[\cC_2]$.
    The order of $[\cC_2]$ is that of $\Sp(1, \Z_d)\ltimes \Z_d^2$, which as shown in \cite{neuhauser2002explicit} is $d^3(d^2-1)$.
    Furthermore, the order of $[\cDb_2\gX]$ is $d^3$, since $[D_2]$ has order $d^2$ and $\gX$ has order $d$.  This implies that $[\cDb_2\gX]$ is a Sylow $d$-subgroup of $[\cC_2]$. All Sylow $d$-subgroups are conjugate, so this means $C$ is Clifford-conjugate to an element of $\cDb_2\gX$.
\end{proof}
\begin{lemma}
    If $DP\in \cDb_{k}\cN$ is of order $d$, then $DP\in \cDb_k\gX$.\label{lemma: DP in pair means its DX}
\end{lemma}
\begin{proof}
    We first show that if $P\in \cN$ has order $d$, then $P\in \gX$.
    Our discussion of the structure of $\cN$ at the beginning of Section \ref{Section: Normal From for Cliffords} shows that $\cN$ has order $d(d-1)$, while $\gX$ has order $d$. Thus $\gX$ is a Sylow $d$-subgroup of $\cN$. As $\gX$ is also a normal subgroup of $\cN$, it follows that $\gX$ is the only Sylow-$d$ subgroup of $\cN$. Hence, any order-$d$ element of $\cN$ is in $\gX$. 
    
    If $DP\in \cDb_k\cN$ has order $d$, note that we may write $(DP)^d=D'P^d$ for some diagonal $D'$. Thus in order for $DP$ to have order $d$, the permutation matrix $P$ must itself have order $d$. But then $P\in \gX$ and $DP\in \cDb_k\gX$.
\end{proof}

\begin{lemma}[Main Lemma 1]
    Let $k \geq 1$ and assume that $\cC_k = \semc_k$: every $k\textsuperscript{th}$\!-level gate is semi-Clifford.  Suppose $G \in \cC_k$ is of order $d$ and $G^2 \in \cC_k$.  Then $G$ is Clifford-conjugate to a gate in $\cDb_k\gX$.
    \label{lemma: Only gates in k-closed pairs are in D_k X}
\end{lemma}
\begin{proof}
    As every Pauli gate is in $\cDb_1$, we can assume that $k \geq 2$.
    
    Since $G$ is semi-Clifford, it is Clifford-conjugate to $DC$ for some $D\in \cD_k$ and $C\in \cC_2$.  

    Without loss of generality, we can assume that $C\not\in\cDb_2\cN$.  Otherwise, if $C \in\cDb_2\cN$, by Lemma \ref{lemma: DP in pair means its DX}, $C \in \cDb_2\gX$ and $DC \in \cDb_k\gX$. 
    
    To finish the proof, it suffices to show that $W = D X D^*$ is a Pauli gate.  If this holds, since $D Z D^* = Z$ is a Pauli gate, $D$ is a Clifford gate.  Thus, $DC$ is a Clifford gate of order $d$ and Lemma \ref{lemma: DX^c is Sylow d-subgroup} establishes our claim.  
        
    Since $G^2 \in \cC_k$, we have that $DCDC\in \cC_k = \semc_k$.  As $C$ is a Clifford gate, it follows that $DCD\in \semc_k$ is semi-Clifford. From $DCD$, we obtain its $(k-1)$-closed conjugate pair 
    \begin{equation}(U,V)=(DCDZD^*C^*D^*, DCDXD^*C^*D^*)=(DCZC^*D^*, DCDXD^*C^*D^*).\end{equation}
    Since $C\not\in\cDb_2\cN$, by Lemma \ref{lemma: DN stabilizes horizontal lines},
    we have that $CZC^*=\omega^cZ^pX^q$, for some $p,q,c\in \Z_d$, where $q\neq 0$. Hence,
    \begin{align}
        U=\omega^c DZ^pX^qD^*=\omega^c Z^p(DX^qD^*),
    \end{align}
    which is a Pauli gate only if $W = (DX^qD^*)^{q^{-1}}$ is a Pauli gate.

    Assume for contradiction that $U$ is not a Pauli gate. The computation above shows that $U \in \cDb_{k-1}\gX$, so the support of $f_U$ is contained in a horizontal line $L_U$ in the phase plane $\Z_d^2$. Since $U$ is not a Pauli gate and is in a conjugate pair, Lemma \ref{obs: Pauli condition on f_V functions} ensures that the Pauli support of $f_U$ has size at least 2. As two points in the phase plane define a unique line, it follows that $L_U$ is the unique line in the phase plane that contains the support of $f_U$. 

    Since $(U, V)$ is a conjugate pair corresponding to the semi-Clifford gate $DCD$, Theorem \ref{cor: Pauli support of conj. pairs} ensures that the Pauli supports of $U$ and $V$ are contained in two parallel lines. Hence, as $L_U$ is horizontal and is the \textit{unique} line containing the Pauli support of $U$, the Pauli support of $V$ must be contained in a horizontal line. By Lemma \ref{lemma: Pauli supports of DX}, $V$ must then itself be in $\cDb_{k-1}\gX$. That is,
\begin{comment}

    By hypothesis, since $(U, V)$ is a $(k-1)$-closed conjugate pair, both $U$ and $V$ are in $C_2\cDb_{k-1}\gX C_2^*$ for some Clifford $C_2$. Hence, the support of $f_{C_2^*UC_2}$ is contained in some horizontal line $L_U'$ in the phase plane. As $C_2^*UC_2$ is not Pauli, we again see that $L_U'$ is the unique line in the phase plane that contains the support of $f_{C_2^*UC_2}$. Since $C_2$ is Clifford, Lemma \ref{lemma: Clifford-conjugate correspondence in f_M function} ensures that $|f_U|=|f_{C_2^*UC_2}\circ S|$ for some $S\in \Sp(1, \Z_d)$. In particular, $f_U$ and $f_{C_2^*UC_2}\circ S$ have the same supports. As $S$ is a linear transformation, it maps lines in the phase space to other lines, so the uniqueness of $L_U$ and $L_U'$ ensures that $L_U=S^{-1}(L_U')$. Thus $S$ maps one horizontal line to another horizontal line, and the linearity of $S$ then ensures that $S$ maps all horizontal lines to horizontal lines.

    Since $V$ is also in $C_2\cDb_{k-1}\gX C_2^*$, the support of $C_2^*VC_2$ is also contained in some horizontal line $L_V'$. Then as above, $|f_V|=|f_{C_2^*VC_2}\circ S|$, so the support of $f_V$ is contained in $S^{-1}(L_V')$. But $S^{-1}(L_V')$ is a horizontal line since $L_V'$ is, so the support of $f_V$ is contained in a horizontal line. Thus $V$ is in fact an element of $D_{k-1}\gX$. That is,
    
\end{comment}
    \begin{align}
        DCDXD^*C^*D^*\in \cDb_{k-1}\gX.
    \end{align}
    As conjugating an element of $\cDb_{k-1}\gX$ by $D^*$ yields another element of $\cDb_{k-1}\gX$, it follows that $CDXD^*C^*\in \cDb_{k-1}\gX$.
    Since $W=DXD^*$ is in $\cDb_k\gX$, the support of $f_W$ is contained in a horizontal line $L_W$ in the phase plane. The support of $f_{CWC^*}$ is then contained in another line $L_W'$ in the phase plane. Since $C\not\in \cDb_2\cN$, Lemma \ref{lemma: DN stabilizes horizontal lines} implies that the symplectic transformation corresponding to $C$ does not carry horizontal lines to horizontal lines, so $L_W'$ is not horizontal. However, $CWC^*$ is in $\cDb_{k-1}\gX$, so its support must also be contained in a horizontal line. Hence, the support of $CWC^*$ is contained in the intersection of $L_W'$ and a horizontal line, which ensures that the support of $f_{CWC^*}$ consists of only one point. As $CWC^*$ is conjugate to $X$, it has order $d$, so it then follows from Lemma \ref{obs: Pauli condition on f_V functions} that $CWC^*$ is Pauli. As $C$ is a Clifford gate, both $W=DXD^*$ and $U$ are then also Pauli gates, which completes the proof.
\begin{comment}

    We may rewrite this as
    \begin{align}
        (CDXD^*X^*C^*)(CXC^*)\in \cD_{k-1}\gX.
    \end{align}
    Since $CXC^*$ is Pauli, it is in $\cD_{k-1}\gX$. As $\cD_{k-1}\gX$ is a group, it follows that
    \begin{align}
        CDXD^*X^*C^*\in \cD_{k-1}\gX.
    \end{align}
\end{comment}  
\end{proof}

This lemma makes an assumption that is equivalent to asserting that $(k-1)$-closed conjugate pairs are \textit{jointly} Clifford-conjugate to a pair of almost diagonal gates.  However, the conclusion proves something weaker than this about a $k$-closed conjugate pair $(U, V)$: $U$ and $V$ are merely both \textit{individually} Clifford-conjugate to almost diagonal gates.  The following lemma resolves this issue and will allow us to inductively establish that $\cC_k = \mathcal{SC}_k$ for all $k$. Interestingly, it does not actually require the inductive hypothesis that the preceding lemma did.

\begin{lemma}[Main Lemma 2]
    Let $(U, V)$ be the conjugate pair of a gate $G \in \cU$ and suppose that both $U$ and $V$ are individually Clifford-conjugate to elements of $\cDb_k\gX$. Then $G$ is simplifiable.
    \label{lemma: Pair individually conjugate to DX is jointly conjugate}
\end{lemma}
\begin{proof}
    Note that if $U$ is a Pauli gate, the claim holds, as we may then conjugate $V$ by a Clifford gate to map it into $\cDb_k\gX$, and this Clifford conjugation sends $U$ to another Pauli gate, which is in $\cDb_k\gX$. By the same reasoning, the claim holds if $V$ is a Pauli gate. We now assume that neither $U$ nor $V$ is a Pauli gate.
    
    By conjugating by an appropriate Clifford gate, we may assume without loss of generality that $U\in \cDb_k\gX$, while $V\in C\cDb_k\gX C^*$ for some Clifford gate $C$. Using our normal form for Clifford gates, either $C\in \cDb_2\cN$ or $C=DPH^*E^c$ for some $DP\in \cDb_2\cN$, $c\in\Z_d$, and $E$ as in Definition \ref{defn: E and M}. Since $\mathcal{D}_2\cN$ normalises $\cDb_k\gX$, if $C\in \cDb_2\cN$ then $V\in \cDb_k\gX$ and we are done. We will show that the other case leads to a contradiction. If $C=DPH^*E^c$, then as conjugating by $E^c$ sends $\cDb_k\gX$ to itself, we have 
    \begin{equation}V\in C\cDb_k\gX C^*=DPH^*\cDb_k\gX(DPH^*)^*.\end{equation}
    Furthermore, as $(DP)^*$ normalises $\cDb_k\gX$, we may conjugate both $U$ and $V$ by $(DP)^*$ to create a conjugate pair with the first gate in $\cDb_k\gX$ and the other in $H^* \cDb_k\gX H$. We now redefine $(U, V)$ to be this conjugate pair. 

    Since $U\in \cDb_k\gX$, the support of $f_U$ is contained in the horizontal line $\{(p,q_0) \ | \ p\in \Z_d\}$ for some $q_0\in \Z_d$. Similarly, as $V\in H^*\cDb_k\gX H$ and conjugation by $H^*$ maps horizontal lines in the phase plane to vertical lines, the support of $f_V$ is contained in the vertical line $\{(p_0, q) \ | \ q\in \Z_d\}$ for some $p_0\in \Z_d$. Using Lemma \ref{lemma: Multiplying f_V functions}, for any $p,q\in \Z_d$ we have
    \begin{align}
        f_{UV}(p,q)=\sum_{i, j\in \Z_d}\omega^{2^{-1}(qi-pj)}f_U(i, j)f_V(p-i, q-j).
    \end{align}
    Since $f_U(i, j)$ vanishes except when $j=q_0$ and $f_V(p-i, q-j)$ vanishes except when $p-i=p_0$, this simplifies to 
    \begin{align}
        f_{UV}(p,q)=\omega^{2^{-1}(q(p-p_0)-pq_0)}f_U(p-p_0, q_0)f_V(p_0, q-q_0).
    \end{align}
    By the same reasoning, we compute that  
    \begin{align}
        f_{VU}(p,q)=\omega^{-2^{-1}(q(p-p_0)-pq_0)}f_U(p-p_0, q_0)f_V(p_0, q-q_0).
    \end{align}
    Then as $UV=\omega VU$, we have $f_{UV}=\omega f_{VU}$. The above equations thus ensure that whenever $f_U(p-p_0, q_0)$ and $f_V(p_0, q-q_0)$ are both nonzero, 
    \begin{align}
        \omega^{2^{-1}(q(p-p_0)-pq_0)}=\omega^{-2^{-1}(q(p-p_0)-pq_0)+1}, 
    \end{align}
    or equivalently, 
    \begin{align}
        pq-pq_0-qp_0-1=0.\label{eq: simultaneous Clifford conjugation of DX gates equation}
    \end{align}
    As $U$ is in a conjugate pair and is not Pauli, Lemma \ref{obs: Pauli condition on f_V functions} ensures that the support of $f_U$ has size at least 2.  We may then choose $p_1\neq p_0$ such that $f_U(p_1-p_0, q_0)\neq 0$. By the same reasoning, since $V$ is not a Pauli gate, we may choose distinct $q_1, q_2$ such that $f_V(p_0, q_1-q_0)$ and $f_V(p_0, q_2-q_0)$ are both nonzero. Then both $(p,q)=(p_1, q_1)$ and $(p,q)=(p_1, q_2)$ must be solutions to  \eqref{eq: simultaneous Clifford conjugation of DX gates equation}. However, for any fixed value of $p$ other than $p=p_0$,  \eqref{eq: simultaneous Clifford conjugation of DX gates equation} is a linear polynomial in $q$, which has exactly one solution. This gives the desired contradiction. 
\end{proof}

We may now prove the main theorem.
\begin{theorem}[Main Theorem]
    For all odd primes $d$, every gate at every level of the single-qudit Clifford hierarchy is semi-Clifford.
    \label{thm: Proof of semi-Clifford conjecture}
\end{theorem}
\begin{proof}
    We claim that for all $k\geq 1$, every gate at level $k$ is semi-Clifford. We proceed by induction on $k$. The base case of $k=1$ holds, as all Pauli gates are semi-Clifford. Assume now that $k\geq 2$ and that the claim holds at level $k-1$. Let $G\in \cC_k$ and $(U, V)$ be the $(k-1)$-closed conjugate pair corresponding to $G$. As $(U, V)$ is $(k-1)$-closed, $U^2$ and $V^2$ are also in $\cC_{k-1}$, so the inductive hypothesis combines with Lemma \ref{lemma: Only gates in k-closed pairs are in D_k X} to ensure that $U$ and $V$ are both individually Clifford-conjugate to elements of $\cDb_{k-1}\langle X\rangle$. Lemma \ref{lemma: Pair individually conjugate to DX is jointly conjugate} then ensures that $(U, V)$ is jointly Clifford-conjugate to a pair of gates in $\cDb_{k-1}\langle X\rangle$. Thus $G$ is simplifiable, and so semi-Clifford by Theorem \ref{cor: Pairs (CD_1X^aC^*, CD_2X^bC^*) correspond to semi-Cliffords}. By induction, every gate in $\cC_k$ is semi-Clifford for $k\geq 1$.
\end{proof}

\subsection{A normal form for Clifford hierarchy gates}\label{Section: Structure of the hierarchy}
Our proof of the conjecture also has several nice consequences for the structure of the hierarchy. First, as the inverse of a semi-Clifford gate in $\cC_k$ is also in $\cC_k$, we have the following corollary.
\begin{cor}
    If $U\in \cC_k$, then $U^*\in\cC_k$.
\end{cor}
Thus, while $\cC_k$ is not a group for $k\geq 3$, it is at least closed under inversion for all $k$. 

The fact that every hierarchy gate is semi-Clifford tells us a great deal about the structure of the hierarchy, as we now know that each hierarchy gate has the form $C_1DC_2$. However, the expression $C_1DC_2$ is highly nonunique. It would be useful to have a unique way to write each gate in the hierarchy. In the theorem below, we extend our Clifford normal form from Lemma \ref{lemma: Normal Form for Clifford Gates} to a normal form for all semi-Clifford gates, and hence, for all gates in the hierarchy by Theorem \ref{thm: Proof of semi-Clifford conjecture}. For $k\geq 2$, we let $\cD_k/\cD_2$ denote a fixed set of representatives for the quotient group $\cD_k/\cD_2$, and we may assume that $\I$ is one of these representatives.  Explicit representatives are easily specified using the classification of diagonal hierarchy gates of Cui-Gottesman-Krishna \cite{cui2017diagonal}.

\begin{theorem}[Normal Form for Semi-Clifford Gates]\label{thm: Normal form for semi-Cliffords}
    Let $k\geq 2$. Then any gate $G$ in $\cC_k$ may be written uniquely in one of the following forms:
    \begin{align}
        G=\begin{cases}
            C & G\in \cC_2\\
            MDC & G\not\in \cC_2 \ \ .
        \end{cases}
    \end{align}
    Here, $\cM$ is the finite set of gates of Definition \ref{defn: E and M}, $D\in(\cD_k/\cD_2)\setminus\{\I\}$, and $C\in \cC_2$.
\end{theorem}
\begin{proof}
    Let $k\geq 2$ and $G\in \cC_k$.  By Theorem \ref{thm: Proof of semi-Clifford conjecture}, $G$ is semi-Clifford. Let $(U, V)$ be the conjugate pair corresponding to $G$. 
    If $U$ and $V$ both have Pauli supports of size 1, Lemma \ref{obs: Pauli condition on f_V functions} implies that $U$ and $V$ are both Pauli, and so that $G$ is a Clifford gate. In this case, we are done. Assume now that $G$ is not a Clifford gate. Then at least one of $U$ and $V$ has a Pauli support of size greater than 1, and we assume without loss of generality that $U$ does. 

    By Corollary \ref{cor: Pauli support of conj. pairs}, as $G$ is semi-Clifford, the Pauli supports of $U$ and $V$ are contained in two parallel lines in $\Z_d^2$. Since $U$ has a Pauli support of size greater than 1, there is then a unique line $L_U$ containing $\supp(f_U)$. By Lemma \ref{lemma: M's correspond to slopes}, there is a unique $M\in \cM$ such that the symplectic transformation $S$ corresponding to $M$ carries horizontal lines to lines parallel to $L_U$. Hence, this is the unique $M$ such that the Pauli supports of $M^*UM$ and $M^*VM$ are contained in two horizontal lines. By Lemma \ref{lemma: DN stabilizes horizontal lines}, this is the unique $M$ such that $M^*G=DC$ for some $D\in \cD_k$ and $C\in \cC_2$.

    Since $G\not\in \cC_2$, we have $D\not \in \cD_2$. We may then write $D$ uniquely as $D_1D_2$, where $D_1\in (\cD_k/\cD_2)\setminus\{\I\}$ and $D_2\in \cD_2$. Then $D_2C=C_2\in \cC_2$ and $DC=D_1C_2$. Note that if we also have $DC=D_1'C_2'$ for some $D_1'\in (\cD_k/\cD_2)\setminus\{\I\}$ and $C_2'\in \cC_2$, then 
    \begin{align}
        (D_1')^*D_1=C_2'C_2^*\in \cC_2,
    \end{align}
    so $(D_1')^*D_1\in \cD_2$. As $\cD_k/\cD_2$ contains a unique representative from each coset of $\cD_2$, it follows that $D_1=D_1'$, and so that $C_2=C_2'$. Thus $G=MD_1C_2$ is in the claimed normal form, and each of $M$, $D_1$, and $C_2$ are unique. 
\end{proof}

If we wished, we could use our Clifford normal form from Lemma \ref{lemma: Normal Form for Clifford Gates} to express the Clifford gate $C$ in our semi-Clifford normal form more explicitly. The normal form thus allows us to understand exactly which gates are at each level of the hierarchy. It also enables us to precisely count the number of gates in $\cC_k$, up to phase.
\begin{cor}\label{cor: size of the hierarchy}
        The number of gates in $\cC_k$, up to phase, is exactly 
    \begin{align}
    |[\cC_k]|=\begin{cases}
        d^2 & k=1,\\
        d^3(d^2-1)(d^{k-1}+d^{k-2}-d) & k\geq 2.
    \end{cases}     
    \end{align}
\end{cor}
\begin{proof}
    Up to phase, each Pauli gate is of the form $W(p,q)$ for a unique $(p,q)\in \Z_d^2$, so $[\cC_1]$ has size $d^2$. For $k\geq 2$, Theorem \ref{thm: Normal form for semi-Cliffords} reduces the problem of determining the size of $[\cC_k]$ to counting the number of options for the gates $M, D$, and $C$ used in the normal form. By definition, there are $d+1$ possible choices for $M$. The characterisation of $\cD_k$ \cite{cui2017diagonal} shows that $[\cD_k]$ contains $d^k$ elements. It follows that $(\cD_k/\cD_2)\setminus\{\I\}$ contains $d^{k-2}-1$ gates, up to phase, and so there are $d^{k-2}-1$ options for $D$. The order of $[\cC_2]$ equals that of $\Sp(1, \Z_d)\ltimes\Z_d^2$, which from \cite{neuhauser2002explicit} is $d^3(d^2-1)$, giving that many choices for $C$. Since every hierarchy gate may be written uniquely up to phase as either $C$ or $MDC$ by Theorem \ref{thm: Normal form for semi-Cliffords}, it follows that 
    \begin{align}
        |[\cC_k]|&=d^3(d^2-1)+(d+1)(d^{k-2}-1)d^3(d^2-1)\\
        &=d^3(d^2-1)(d^{k-1}+d^{k-2}-d).\qedhere
    \end{align}
\end{proof}

\section{Conclusions}

This work significantly advances the program of classifying gates of the Clifford hierarchy and semi-Clifford gates.  We have solved, in the practically-relevant case of one qubit/qudit, the problem of specifying precisely the gates of the Clifford hierarchy.  Contributions to the classification problem have been made by many researchers, including founding figures of the field of quantum computing, over the past twenty-five years.

In future work, we will extend our classification of two-qubit/qudit hierarchy gates.  

\begin{conj}[Classification of two-qubit/qudit hierarchy gates]Every hierarchy gate of two qudits is either a Clifford gate or can be \emph{uniquely} expressed as $M_1 M_2 D C$.  Here, $M_1$ is $H_1$ or $H_1D_1H_1^*$ and $M_2$ is $H_2$ or $H_1H_2D_2H_2^*H_1^*$ with $H_i$ being the Hadamard gate on the $i^\text{th}$ qudit, $D_1 = (E_1)^a$, $D_2=(CZ)^b(E_2)^c$ for $a,b,c \in \Z_d$.  The gates $D, C$ are as in Theorem \ref{thm: Normal form for semi-Cliffords}.  In the qubit case, $D_1 = (S_1)^a$ and $D_2 = (CZ)^b (S_2)^c$ for $a,b,c \in \Z_2$.
\end{conj}

This conjectured classification can be extended to semi-Clifford gates of more than two qubits/qudits.  We intend also to give a complete classification for hierarchy gates of any number of qubits/qudits.

Deeper understanding of the Clifford hierarchy and semi-Clifford gates will lead to progress in error correction and fault-tolerance  and more efficient circuit and gate synthesis (e.g \cite{glaudell2022qutrit, jochym2021four, krishna2019towards, webster2023transversal, yeh2023scaling}).  It is also relevant to the theoretical study of quantum advantage.  Hierarchy gates and their magic states are a key resource for powering quantum computation and hierarchy level provides a natural ranking or measure of their power/complexity which can be in turn correlated with advantage (e.g. \cite{de2024quantum, frembs2025no, frembs2023hierarchies}).  For this reason, the Clifford hierarchy is relevant to many seemingly unconnected topics such as classical simulation \cite{bu2019efficient} and learning circuits \cite{anshu2024survey}.  By clarifying the algebraic structure of magic states, we contribute to the understanding of the resources that power quantum computation.  

We further bolster the viability of qudit-based fault-tolerant universal quantum computers by providing many more pathways towards efficient fault-tolerant qudit quantum computing.  This is practically important as qudit magic state distillation has been proposed as a significantly more efficient alternative to the qubit case \cite{quditmsd}.  The advantages of qudit-based computation \cite{tonchev2016quantum} has led to rapidly accelerating development by experimentalists \cite{chi2022programmable,chizzini2022molecular,karacsony2023efficient,low2020practical,ringbauer2022universal,seifert2022time,wang2018proof}.

\section{Acknowledgments}
ND acknowledges support from the Canada Research Chair program, NSERC Discovery Grant RGPIN-2022-03103, the NSERC-European Commission project FoQaCiA, and the Faculty of Science of Simon Fraser University. OL was supported by NSERC Undergraduate Student Research Awards (USRA). 

\begin{comment}
    
\section{Declarations}
The authors have no relevant financial or nonfinancial interests or competing interests to declare that are relevant to the content of this article.

\end{comment}

\bibliographystyle{abbrv}
{\footnotesize
\bibliography{cliff}}

\end{document}